%% file: main.tex
\documentclass[10pt]{article}

\usepackage{amsthm}
\usepackage{graphicx} % support the \includegraphics command and options
\usepackage{array} % for better arrays (eg matrices) in maths

\usepackage{amsmath, amssymb, amsfonts, verbatim}
\usepackage{hyphenat,epsfig,subcaption,multirow}
\usepackage{nicefrac}
\usepackage{paralist}
\usepackage{thm-restate}
\usepackage{mleftright}

\usepackage[font=small, labelfont=bf]{caption}

\usepackage[usenames,dvipsnames]{xcolor}
\usepackage[ruled]{algorithm2e}

\DeclareFontFamily{U}{mathx}{\hyphenchar\font45}
\DeclareFontShape{U}{mathx}{m}{n}{
      <5> <6> <7> <8> <9> <10>
      <10.95> <12> <14.4> <17.28> <20.74> <24.88>
      mathx10
      }{}
\DeclareSymbolFont{mathx}{U}{mathx}{m}{n}
\DeclareMathSymbol{\bigtimes}{1}{mathx}{"91}

\usepackage{tcolorbox}
\tcbuselibrary{skins,breakable}
\tcbset{enhanced jigsaw}

\usepackage[normalem]{ulem}
\usepackage[compact]{titlesec}

\definecolor{DarkRed}{rgb}{0.5,0.1,0.1}
\definecolor{DarkBlue}{rgb}{0.1,0.1,0.5}

\usepackage{nameref}
\definecolor{ForestGreen}{rgb}{0.1333,0.5451,0.1333}
\definecolor{Red}{rgb}{0.9,0,0}
\usepackage[linktocpage=true,
	pagebackref=true,colorlinks,
		urlcolor=black,
	linkcolor=DarkRed,citecolor=ForestGreen,
	bookmarks,bookmarksopen,bookmarksnumbered]
	{hyperref}
\usepackage[noabbrev,nameinlink]{cleveref}
\crefname{property}{property}{Property}
\creflabelformat{property}{(#1)#2#3}
\crefname{equation}{eq}{Eq}
\creflabelformat{equation}{(#1)#2#3}

\usepackage{bm}
\usepackage{url}
\usepackage{xspace}
\usepackage[mathscr]{euscript}

\usepackage{tikz}
\usetikzlibrary{arrows}
\usetikzlibrary{arrows.meta}
\usetikzlibrary{shapes}
\usetikzlibrary{backgrounds}
\usetikzlibrary{positioning}
\usetikzlibrary{decorations.markings}
\usetikzlibrary{patterns}
\usetikzlibrary{calc}
\usetikzlibrary{fit}
\usetikzlibrary{snakes}

\usepackage{mdframed}

\usepackage[noend]{algpseudocode}
\makeatletter
\def\BState{\State\hskip-\ALG@thistlm}
\makeatother

\usepackage{cite}
\usepackage{enumitem}

\usepackage[margin=1in]{geometry}

\newtheorem{theorem}{Theorem}
\newtheorem{lemma}{Lemma}[section]

\newtheorem{proposition}[lemma]{Proposition}
\newtheorem{corollary}[lemma]{Corollary}
\newtheorem{claim}[lemma]{Claim}

\newtheorem*{claim*}{Claim}
\newtheorem*{proposition*}{Proposition}
\newtheorem*{lemma*}{Lemma}
\newtheorem*{problem*}{Problem}

\crefname{lemma}{Lemma}{Lemmas}
\crefname{claim}{Claim}{Claims}

\newtheorem{mdresult}{Result}
\newenvironment{result}{\begin{mdframed}[backgroundcolor=lightgray!40,topline=false,rightline=
		false,leftline=false,bottomline=false,innertopmargin=2pt]
		\begin{mdresult}}{\end{mdresult}\end{mdframed}}

\newtheorem*{mdresult*}{Main Result}

\newtheorem{definition}[lemma]{Definition}
\theoremstyle{definition}

\newtheorem{mdinvariant}[lemma]{Lemma}

\theoremstyle{definition}
\newtheorem{mdalg}{Algorithm}
\newenvironment{Algorithm}{\begin{tbox}\begin{mdalg}}{\end{mdalg}\end{tbox}}

\allowdisplaybreaks

\renewcommand{\qed}{\nobreak \ifvmode \relax \else
      \ifdim\lastskip<1.5em \hskip-\lastskip
      \hskip1.5em plus0em minus0.5em \fi \nobreak
      \vrule height0.75em width0.5em depth0.25em\fi}

\renewcommand{\leq}{\leqslant}
\renewcommand{\geq}{\geqslant}

\input{macros}

\title{Tight Bounds for Vertex Connectivity in Dynamic Streams}
\author{}
\author{
	Sepehr Assadi\footnote{(\href{mailto:sepehr.assadi@rutgers.edu}{sepehr@assadi.info)} 
Department of Computer Science, Rutgers University.  Research supported in part by a NSF 
CAREER 
Grant CCF-2047061, a Google Research gift, and a Fulcrum award from Rutgers Research Council.} \and 
Vihan Shah\footnote{(\href{mailto:vihan.shah98@rutgers.edu}{\text{vihan.shah98@rutgers.edu}}) 
Department of Computer Science, Rutgers University.  Research supported in part by a NSF 
CAREER 
Grant CCF-2047061. }  
}

\date{}

\begin{document}
\maketitle

%%\vspace{-0.55cm}
\thispagestyle{empty}
\pagenumbering{roman}

\input{abstract}
%
%
%\clearpage
%
\bigskip
\setcounter{tocdepth}{3}
\tableofcontents

\clearpage

\pagenumbering{arabic}
\setcounter{page}{1}

\input{intro}

\input{prelim}

\input{certificate}

\input{alg}

\input{LowerBound}

\subsection*{Acknowledgements}
We are grateful to Zachary Langley and Michael Saks for helpful conversations about certificates of vertex connectivity and to the anonymous reviewers of SOSA 2023 for their useful comments on the presentation of the paper. 

\bibliographystyle{alpha}
\bibliography{new}

\bigskip

\appendix

\input{ins_only}

\input{mader}

\end{document}

%% file: macros.tex
\newcommand{\Ot}{\ensuremath{\widetilde{O}}}
\newcommand{\eps}{\ensuremath{\varepsilon}}
\newcommand{\Paren}[1]{\Big(#1\Big)}

\newcommand{\bracket}[1]{\left[#1\right]}
\newcommand{\paren}[1]{\ensuremath{\mleft(#1\mright)}\xspace}
\newcommand{\card}[1]{\left\vert{#1}\right\vert}

\newcommand{\prob}[1]{\Pr\paren{#1}}
\newcommand{\expect}[1]{\Exp\bracket{#1}}

\newcommand{\set}[1]{\ensuremath{\left\{ #1 \right\}}}
\newcommand{\poly}{\mbox{\rm poly}}

\DeclareMathOperator*{\Exp}{\ensuremath{{\mathbb{E}}}}
\DeclareMathOperator*{\Prob}{\ensuremath{\textnormal{Pr}}}
\renewcommand{\Pr}{\Prob}

\newenvironment{tbox}{\begin{tcolorbox}[
		enlarge top by=5pt,
		enlarge bottom by=5pt,
		 breakable,
		 boxsep=0pt,
                  left=4pt,
                  right=4pt,
                  top=10pt,
                  arc=0pt,
                  boxrule=1pt,toprule=1pt,
                  colback=white
                  ]%%
	}
{\end{tcolorbox}}

\newcommand{\II}{\ensuremath{\mathbb{I}}}

\newcommand{\mireal}[1][]{
  \ifx\relax#1\relax%
    \II(\mione \,; \mitwo)%
  \else%
    \II(\mione \,; \mitwo\mid #1)%
  \fi
}

\newcommand{\disj}{\ensuremath{\textnormal{\textsc{DISJ}}_N}\xspace}

\newcommand{\reps}{r}

%% file: abstract.tex
% !TeX root = main.tex 
%!TEX root = main.tex

\begin{abstract}
	We present a streaming algorithm for the vertex connectivity problem in dynamic streams with a 
	(nearly) optimal space bound: for any $n$-vertex graph $G$ and any integer $k \geq 1$, 
	our algorithm  with high probability outputs whether or not $G$ is $k$-vertex-connected in a single 
	pass using $\Ot(k n)$ space\footnote{Throughout the paper, we use $\Ot(f) := O(f \cdot 
	\poly\log{f})$ to hide poly-logarithmic factors.}.

	\smallskip
	
	Our upper bound matches the known $\Omega(k  n)$ lower bound for this problem even in insertion-only streams---which we extend to multi-pass algorithms in this paper---and closes one of the last remaining gaps  
	in our understanding of dynamic versus insertion-only streams. Our result is obtained via a novel analysis of the previous best dynamic streaming algorithm of Guha, McGregor, and Tench [PODS 2015] who obtained an $\Ot(k^2 n)$ 
	space algorithm for this problem. This also gives a model-independent algorithm for computing a ``certificate'' of $k$-vertex-connectivity as a union of $O(k^2\log{n})$ spanning forests, each on a random subset of $O(n/k)$ vertices, which
	may be of independent interest.
	
\end{abstract}

%% file: intro.tex
% !TeX root = main.tex 
%!TEX root = main.tex

\section{Introduction}\label{sec:intro} 

The \emph{vertex connectivity} of an undirected graph $G=(V,E)$, with $n$ vertices and $m$ edges, is the 
size of the smallest vertex cut in $G$, defined as the minimum number of vertices whose removal 
disconnects the graph (or turns it into a singleton vertex).
Finding the vertex connectivity of a graph is a fundamental problem in combinatorial optimization 
and has been extensively studied in the literature; see, e.g.,\cite{kleitman1969methods,podderyugin1973algorithm,even1975network,
	becker1982probabilistic,linial1988rubber,henzinger2000computing,li2021vertex}.
This problem can be solved in ``polylogarithmic max-flow time'' via a result of \cite{li2021vertex}, which combined 
with the recent breakthrough improvement for max-flow computation in \cite{chen2022maximum}, leads to 
an $m^{1+o(1)}$ time algorithm for vertex connectivity.

We study the vertex connectivity problem in \emph{dynamic streams}. 
In this model, the edges of the input graph $G$ are presented to the algorithm as a sequence of 
both edge insertions and deletions. The goal is to, given an integer $k$ at the start of the stream, process the stream with limited space and at the end output whether or not the graph is \emph{$k$-vertex-connected}, namely, its vertex connectivity is at least $k$. 
In fact, in our algorithm, we focus on not only deciding if the graph is $k$-vertex connected or not, but also outputting a \emph{certificate} of $k$-vertex-connectivity defined as follows: 

\begin{definition}\label{def:certificate}
	For any graph $G=(V,E)$, a \textnormal{\textbf{certificate}} of $k$-vertex-connectivity for $G$ is a subgraph on the 
	same vertex set $H=(V,E_H)$ such that $G$ is $k$-vertex-connected if and only if $H$ is 
	$k$-vertex-connected.
\end{definition}

A certificate of $k$-vertex-connectivity needs $\Omega(kn)$ edges since any 
$k$-vertex-connected graph has at least $kn/2$ edges, as it needs to have a minimum degree of 
at least $k$. Mader's theorem (see \Cref{prop:mader-thm}) implies that
there is also always a certificate with $O(kn)$ edges although we will not use this result directly in our dynamic streaming algorithm.

In insertion-only streams with no edge deletions, it has been known since the introduction of the model in~\cite{FeigenbaumKMSZ05}, that one can find a certificate of $k$-vertex connectivity in $\Ot(kn)$ space 
using the sparsification techniques of~\cite{cheriyan1993scan} 
or~\cite{eppstein1997sparsification}\footnote{While this connection has 
been observed in several places, e.g., in~\cite{sun2015tight,guha2015vertex}, we are not aware of a 
concrete reference that includes this proof in the streaming model and thus provide a  
self-contained proof in~\Cref{app:insertion-only} for completeness.}, which is nearly optimal. 
Moreover,~\cite{sun2015tight} proved that $\Omega(kn)$ space is needed even for the original (decision) problem, thus settling the space complexity of the problem in insertion-only streams, up to logarithmic factors. 
But when it comes to dynamic streams, the best upper bound achieves $\Ot(k^2 n)$ space~\cite{guha2015vertex} with no better known lower bounds. We close this gap in this paper. 

Our main result is a general (model-independent) approach for computing a vertex connectivity certificate. 

\begin{result}[Formalized in~\Cref{thm:k-con-cert}] \label{res:main}
	For any graph $G$ and integer $k \geq 1$, let $H$ be a subgraph of $G$ with $O(kn\log{n})$ edges obtained as a union of $O(k^2\log{n})$ spanning forests, each on a random subset of $O(n/k)$ vertices chosen independently of the others. 
	Then, with high probability, $H$ is a certificate of $k$-vertex-connectivity for $G$. This certificate also preserves \underline{all} vertex cuts of size up to $k$ and can determine if any two 
	given vertices are $k$-vertex-connected or not.
\end{result}

\Cref{res:main} reduces the problem of finding a certificate of $k$-vertex-connectivity to computing spanning forests on random subsets of vertices. 
This gives a general approach for solving $k$-vertex-connectivity across different models that can be of its own independent interest. In particular,~\Cref{res:main} immediately leads to a dynamic streaming algorithms for $k$-vertex-connectivity when combined with the dynamic streaming algorithms of~\cite{AhnGM12} for computing spanning forests. 

\begin{mdresult}[Formalized in~\Cref{thm:one-pass-con-alg}]\label{res:alg}
	There is a randomized single-pass $\Ot(kn)$ space algorithm that solves $k$-vertex-connectivity with high probability in dynamic streams.
\end{mdresult}

\Cref{res:alg} also closes one of the last remaining gaps in  understanding of dynamic versus 
insertion-only streams. Starting with the breakthrough of~\cite{AhnGM12} that initiated the study of 
dynamic graph streams, various 
graph problems such as cut sparsifiers \cite{AhnGM12b}, spectral sparsifiers 
\cite{kapralov2017single}, densest subgraph \cite{McGregorTVV15}, 
subgraph counting \cite{AhnGM12b}, 
and $(\Delta+1)$-vertex coloring \cite{AssadiCK19a} were shown to admit algorithms in dynamic streams with similar guarantees 
as those of insertion-only streams. In particular, for the closely related problem of $k$-edge-connectivity, it was shown already 
by~\cite{AhnGM12} how to obtain similar bounds in dynamic streams as insertion-only streams. 
For a few other problems such as maximum matchings and minimum vertex cover, strong separations 
between the two models were proven in~\cite{Konrad15,AssadiKLY16} (see 
also~\cite{DarkK20,AssadiS22, DBLP:conf/approx/NaiduS22}); a {conjectured} separation for the 
shortest path problem also appears 
in~\cite{FiltserKN21}. Before our~\Cref{res:alg} however, it was not clear vertex connectivity belongs to which family of these problems. 

As a secondary result, we also extend the previous $\Omega(kn)$ lower bound of~\cite{sun2015tight} to multiple passes.

\begin{mdresult}[Formalized in~\Cref{thm:LB-multipass}]\label{res:lower}
	Any randomized $p$-pass streaming algorithm that solves $k$-vertex-connectivity 
	with constant probability even in insertion-only streams, needs $\Omega(kn/p)$ space.
\end{mdresult}
%\vihan{Why is the environment different than result 1}

\Cref{res:lower} is proven by lower bounding the communication complexity of $k$-vertex-connectivity problem with $\Omega(kn)$ bits. This also answers an open problem of 
\cite{blikstad2022nearly} in negative that asks whether recent equivalences between vertex connectivity and vertex-capacitated max-flow in~\cite{li2021vertex} in classical setting also extends to communication complexity model 
(a recent result of~\cite{blikstad2022nearly} shows that vertex-capacitated max-flow (unit capacity) can be solved with $\Ot(n)$ communication, while our communication lower bound in general implies an $\Omega(n^2)$ lower 
bound for determining the vertex connectivity of a graph, thus ruling out such an equivalence).  We remark that our~\Cref{res:lower}, similar to the previous lower bound for streaming vertex connectivity in~\cite{sun2015tight}, holds on multi-graphs (with at most two parallel edges per pairs of vertices); our algorithm in~\Cref{res:alg}
also works for multi-graphs. It remains an interesting open question to prove any streaming lower bound for vertex connectivity on simple graphs as well (even for single-pass algorithms). 

\paragraph{Our techniques.} \Cref{res:main} (and by extension~\Cref{res:alg}) is obtained via a novel and improved \emph{analysis} of the dynamic streaming algorithm of~\cite{guha2015vertex} (with minor modifications), that leads to an improved 
bound on the size of the certificate\footnote{The algorithm in \cite{guha2015vertex} is presented as a streaming 
algorithm but it implicitly gives a certificate with $O(k^2n\log{n})$ edges.}.~\cite{guha2015vertex} presented algorithms for two \emph{relaxations} of the vertex connectivity problem in dynamic streams:  
\begin{itemize}
	\item $k$-vertex-\textbf{query}-connectivity problem: the algorithm is additionally given a set $X$ of size at most $k$ \emph{after} the stream and the goal is to determine 
	whether removing $X$ from $G$ disconnects the graph or not (this is similar to the vertex-failure connectivity oracle problem studied in\cite{long2022near}). \cite{guha2015vertex} 
	showed that this problem can be solved in $\Ot(kn)$ space (and that this is also nearly optimal for this problem). 
	
	\item \textbf{promised-gap} $k$-vertex-connectivity problem: the algorithm is given a parameter $\eps \in (0,1)$ at the \emph{start} of the stream and the goal is to determine 
	whether the vertex connectivity of the input graph $G$ is at least $k$ or at most $(1-\eps) \cdot 
	k$. \cite{guha2015vertex} showed that this problem could be solved in $\Ot(kn/\eps)$ space 
	(using an algorithm quite 
	similar to the one for the previous case). 
\end{itemize}
Both these algorithms work roughly as follows (think of $\eps = \Theta(1)$ for the second one in this context): for $\Ot(k^2)$ times in parallel, sample $\Ot(n/k)$ vertices from the input graph uniformly at random and maintain a spanning forest on these vertices using the dynamic streaming algorithm of~\cite{AhnGM12} (this is basically the same approach taken in our~\Cref{res:main}). 
Then, solve the problem on these stored set of edges at the end of the stream. This approach can be 
implemented in $\Ot(k^2) \cdot \Ot(n/k) = \Ot(kn)$ space and~\cite{guha2015vertex} proves, using a 
somewhat different analysis, that this solves the problem in each case with high probability. 

One can solve the original $k$-vertex-connectivity problem using this approach as follows. For the query problem, boost the probability of success of the algorithm to $1-n^{-k}$ by running the algorithm in parallel $\Theta(k\log{n})$ 
times; this allows for taking a union bound over at most $\binom{n}{k}$ possible choices for the 
query set $X$ and testing whether removal of \emph{any} of them can disconnect the graph. For the 
promised-gap problem, we can  
set $\eps = 1/k$ which allows us to distinguish between graphs with vertex connectivity $k$ versus $k-1$ and thus solve the $k$-vertex-connectivity problem. Nevertheless, as is apparent, either of these solutions leads to an algorithm with  $\Ot(k^2n)$ space
which is sub-optimal. 

The key novelty in our work is another analysis of essentially the same vertex-sampling plus spanning forest computation approach of~\cite{guha2015vertex}. This analysis allows us to ``beat the union bound'' over all possible $k$-subsets of vertices
as candidate choices for the vertex cut, that was the source of the additional factor $k$ in space in the algorithm of~\cite{guha2015vertex}. In particular, our analysis consists of two parts. We first show that all pairs of vertices 
with sufficiently ``high'' vertex connectivity, say, at least $2k$, remain at least 
$k$-vertex-connected even over the stored edges of the sampling approach (this part is quite 
similar to the guarantee of promised-gap algorithm of~\cite{guha2015vertex}). 
We then prove that \emph{all} edges in the input graph with ``low'' vertex connectivity between their endpoints, say, less than $2k$, are recovered by this sampling approach entirely\footnote{The fact that the number of these edges itself is sufficiently small is a direct corollary of Mader's theorem (\Cref{prop:mader-thm}) that 
states that every graph with $O(kn)$ edges has a $(2k)$-vertex-connected subgraph. This theorem is at the heart of existing (near) optimal algorithms for vertex connectivity in insertion-only streams (see~\Cref{app:insertion-only}).
Nevertheless, since our proof requires additionally recovering these edges via a particular sampling method, it does \emph{not} rely on Mader's theorem, and instead, as a corollary, implies a weaker variant of Mader's theorem (with $O(kn\log{n})$ edges instead) via a probabilistic argument quite different from the standard proofs of this theorem (\Cref{app:Mader}).}. Finally, we combine these two parts to argue 
that the sampled set of edges is a certificate for $k$-vertex-connectivity of the input graph and conclude the proof. This proof more generally shows that \emph{all} minimum (global) vertex cuts as well as all $s$-$t$ vertex cuts of size up to $k$ are 
preserved in this sampling process.

%% file: prelim.tex
% !TeX root = main.tex 
%!TEX root = main.tex

\section{Preliminaries}\label{sec:prelim}

 \paragraph{Notation.}
For a graph $G=(V,E)$, we use $\deg(v)$ and $N(v)$ for each vertex $v \in V$ to denote the degree 
and neighborhood of $v$, respectively. For a subset $F$ of edges in $E$, we use $V(F)$ to denote 
the vertices incident on $F$; similarly, for a set $U$ of vertices, $E(U)$ denotes the edges 
incident on $U$. We further use $G[U]$ for any set $U$ of vertices to denote the induced subgraph of $G$ on $U$. 
For any two vertices $s,t \in V$, we say that a collection of $s$-$t$ paths are vertex-disjoint if they do not share any vertices other than $s$ and $t$.

We use the following standard forms of Chernoff bounds.

\begin{proposition}[Chernoff bound; c.f.~\cite{dubhashi2009concentration}]\label{prop:chernoff}
	Suppose $X_1,\ldots,X_m$ are $m$ independent random variables with range $[0,b]$ each for some $b \geq 1$. Let 
	$X := \sum_{i=1}^m X_i$ and $\mu_L \leq \expect{X} \leq \mu_H$. Then, for any $\eps > 0$, 
	\[
	\Pr\paren{X >  (1+\eps) \cdot \mu_H} \leq \exp\paren{-\frac{\eps^2 \cdot \mu_H}{(3+\eps) \cdot b}} \quad 
	\textnormal{and} \quad \Pr\paren{X <  (1-\eps) \cdot \mu_L} \leq \exp\paren{-\frac{\eps^2 \cdot 
	\mu_L}{(2+\eps) \cdot b}}.
	\]
\end{proposition}

We use the term ``with high probability" to mean with probability at least 
$1-1/n^c$ for some large constant $c > 0$, which can be made arbitrarily large by  
increasing the space of our algorithms with a constant factor.

 \paragraph{Dynamic graph streams.}
The dynamic graph streaming model is defined formally as follows. 
 
\begin{definition}\label{def:dynamic-stream}
	A dynamic stream $\sigma = (\sigma_1,\ldots,\sigma_N)$ defines a multi-graph $G=(V, E)$ on $n$ 
	vertices.
	Each entry of the stream is a tuple $\sigma_k = (i_k, j_k,\Delta_k)$ for $i_k, j_k \in [n]$ and 
	$\Delta_i \in \set{-1,+1}$. 
	The multiplicity of an edge $(u,v)$ is defined as:
	\[
		A(u,v) = \sum_{\sigma_k: i_k = u \; \wedge \; j_k=v} \Delta_k. 
	\]
	The multiplicity of every edge is required to be always non-negative.
\end{definition}

The goal in this model is to design algorithms that can process a dynamic stream using limited space 
and at the end of the stream, output a solution to the underlying problem for the (multi-)graph 
defined
by the stream. Throughout the paper, we measure the space of the algorithms in \emph{bits}. 

We use the algorithm of~\cite{AhnGM12} that can find a spanning forest of a graph in a dynamic 
stream. 
\begin{proposition}[\!\!\cite{AhnGM12}] \label{prop:spanning-forest}
	There is an algorithm that given any $N$-vertex graph $G$ in a dynamic stream and  
	$\delta \in (0,1)$, computes a spanning forest $T$ of $G$ with probability at least $1-\delta$ in 
	$O(N \log^3 (N/\delta))$ space.
\end{proposition}
\vspace{5pt}

%% file: certificate.tex
% !TeX root = main.tex 
%!TEX root = main.tex

\section{A Certificate of Vertex Connectivity}\label{sec:Certificate}

We present a certificate of $k$-vertex-connectivity in this section, formalizing~\Cref{res:main}. 
Our algorithm is virtually identical (up to changing constants and ignoring implementation 
details in dynamic streams) to the algorithm in \cite{guha2015vertex}
for the $k$-vertex-\textbf{query}-connectivity problem mentioned in the introduction. 
However, we provide an improved analysis showing that it also works for the $k$-vertex-connectivity problem.

\begin{Algorithm}\label{alg:cert}
	An algorithm for computing a certificate of $k$-vertex-connectivity.
	
	\medskip
	
	\textbf{Input:} A graph $G = (V,E)$ and an integer $k$.
	
	\medskip
	
	\textbf{Output:} A certificate $H$ for $k$-vertex-connectivity of $G$.

	\smallskip
		
	\begin{enumerate}
		\item For $i = 1, 2, \ldots, \reps := \paren{200 k^2 \ln n}$ do the following:
		\begin{itemize}
			\item Let $V_i$ be a subset of $V$ where each vertex is sampled independently with 
			probability $1/k$.
			\item Let $G_i=G[V_i]$ be the induced subgraph of $G$ on $V_i$.
			\item Compute a spanning forest $T_i$ of $G_i$. 
		\end{itemize}
	\item Output $H:= T_1 \cup T_2 \cup \ldots \cup T_{\reps}$ as a certificate for $k$-vertex-connectivity 
	of $G$.
\end{enumerate}
	
\end{Algorithm}

The following theorem proves the main guarantee of this algorithm. 

\begin{theorem}\label{thm:k-con-cert}
	\Cref{alg:cert}, given any graph $G=(V,E)$ and any integer $k \geq 1$, outputs a 
	certificate $H$ of $k$-vertex-connectivity of $G$ with $O(k n \cdot \log n)$ edges with high 
	probability.
\end{theorem}

The analysis in the proof of~\Cref{thm:k-con-cert} is twofold.
We first show that pairs of vertices that are at least $2k$-vertex-connected in $G$ stay $k$-vertex-connected in 
$H$. 
Secondly, we show that edges whose endpoints are not $2k$-vertex-connected in $G$ will be 
preserved in $H$.
Putting these together, we then show that $H$ is a certificate for $k$-vertex-connectivity of $G$ 
and 
has at most $\Ot(kn)$ edges.

We start by bounding the number of edges of the certificate $H$.
We first show that the sum of sizes of $V_i$ is $O(kn \log n)$ with high probability.
\begin{claim}\label{clm:cert-space-help}
	$\sum_{i=1}^{r} \card{V_i} = O(kn \cdot \log n)$ with high probability.
\end{claim}
\begin{proof}
For any iteration $i \in [r]$, the graph $G_i$ has $n/k$ vertices in expectation.
We have $\reps=O(k^2 \ln n)$ iterations so $\sum_{i=1}^r \card{V_i}= O(k n 
\cdot \log n)$ in expectation. We  prove this is the case with high probability as well. 
	
For $i \in [\reps]$, let $X_i$ be the random variable denoting the number of vertices in $V_i$.
We know that $0 \leq X_i \leq n$ (the inequalities are tight when $V_i=\emptyset$ 
and $V_i=V$).
Let $X=\sum_i X_i$ be the random variable governing the sum of sizes of $V_i$'s.
We have $\expect{X_i} = n/k$ implying $\expect{X}= \mu=r \cdot (n/k)$.
Using a Chernoff bound (\Cref{prop:chernoff}) with parameters $b = n$ and $\eps = 1$ we get:
\begin{align*}
	\prob{X > 2 \mu} &\leq \exp\paren{\frac{ -\mu}{4 b}} 
	=\exp\paren{\frac{ -\reps \cdot (n/k) }{4 n}} 
	= \exp\paren{-50k\ln{n}} \leq n^{-50}.
\end{align*}	
Thus, we get that the $\sum_{i=1}^r \card{V_i} \leq 2\mu_H = O(kn \cdot \log n)$ with high 
probability as well. 
\end{proof}

\begin{lemma}\label{lem:cert-space}
	The certificate $H$ in \Cref{alg:cert} has $O(k n \cdot \log n)$ edges with high probability.
\end{lemma}
\begin{proof}
	Each spanning forest $T_i$ has at most $\card{V_i}$ edges.
	Thus, the total number of edges in $H$ can be bounded by $\sum_{i=1}^r \card{V_i} = O(kn \cdot \log n)$ 
	with high probability (by \Cref{clm:cert-space-help}).
\end{proof}

We now prove the correctness of this algorithm in the following lemma. 
\begin{lemma}\label{lem:cert-corr}
	Subgraph $H$ of~\Cref{alg:cert} is a certificate of $k$-vertex-connectivity for $G$ with high 
	probability.
\end{lemma}

\Cref{lem:cert-corr} will be proven in two steps. We first show that every pair of vertices that have at least 
$2k$ vertex-disjoint paths 
between them in $G$ have at least $k$ vertex-disjoint paths in $H$ with high 
probability\footnote{By Menger's theorem (\Cref{prop:menger}), this is equivalent to saying any pair 
of vertices that are $(2k)$-vertex-connected in $G$ remain at least $k$-vertex-connected in $H$. 
However, we do not need to explicitly use Menger's theorem in our proofs.}.
\begin{lemma}\label{clm:connectivity-in-H}
	Every pair of vertices $s,t$ in $G$ that have at least $2k$ vertex-disjoint paths between them in $G$
	have at least $k$ vertex-disjoint paths in $H$ with high probability.
\end{lemma}

We then show that every edge whose endpoints have less than $2k$ vertex-disjoint 
paths between them in $G$ will belong to $H$ as well. 

\begin{lemma}\label{clm:edges-in-H}
	Every edge $(s,t) \in G$ that has less than $2k$ vertex-disjoint paths between its endpoints in 
	$G$ belongs to $H$ also with high probability.
\end{lemma}
The proofs of these lemmas appear in the next two subsections. We first use these lemmas to prove 
\Cref{lem:cert-corr} and conclude the proof of~\Cref{thm:k-con-cert}.

\begin{proof}[Proof of \Cref{lem:cert-corr}]
	We first condition on the events in \Cref{clm:connectivity-in-H} and 
	\Cref{clm:edges-in-H} both of which happen with high probability.
	 We also condition on the event that $H$ has the same set of vertices 
	as $G$. This event also happens with high probability because the probability that a given vertex is not in $H$ is $(1 -1/k)^\reps \leq 
	\exp(-200k 
	\cdot \ln n)= n^{-200k}$ and hence by a union bound, all vertices in $G$ are 
	also in $H$ with high probability. All in all, by a union bound, all the above events 
	happen together with high probability. 
	
	We need to show that $H$ is $k$-vertex-connected iff $G$ is $k$-vertex-connected.
	If $H$ is $k$-vertex-connected then $G$ is also $k$-vertex-connected simply because $H$ is a 
	subgraph of $G$ (and crucially on the same set of vertices, namely, it is a spanning subgraph). 
	
	We now assume towards a contradiction that $G$ is $k$-vertex-connected, but $H$ is not.
	This means that there is a vertex cut $X$ of size at most $k-1$ such that there is a partition $(S,X,T)$ of 
	$V$ with no edges between $S$ and $T$ in $H$ (so that removing $X$ disconnects $H$).
	Since $G$ is $k$-vertex-connected, $X$ cannot be a vertex cut in $G$ and thus $G$ has an edge $e=(s,t)$ between $S$ and $T$ (see~\Cref{fig:partition}). 
	
	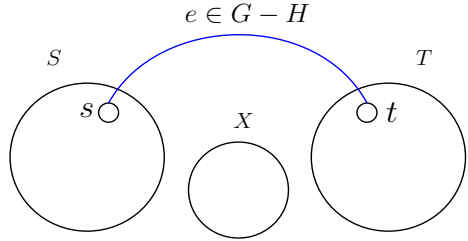
\begin{figure}[h!]
	\centering
	{   \resizebox{180pt}{90pt}{ \input{partition}}  }
	\caption{An illustration of the partition $(S,X,T)$ of vertices of $G$ and $H$. There are no edges  between $S$ and $T$ in $H$, 
	while $G$ has at least one edge $e=(s,t)$ between $S$ and $T$, to ensure its $k$-vertex-connectivity as $\card{X} < k$. 
		\label{fig:partition}}
\end{figure}
	\noindent
	We now consider two cases.
	\begin{itemize}
		\item Case 1: $s$ and $t$ have at least $2k$ vertex-disjoint paths between them in $G$. \\
		By conditioning on the event of \Cref{clm:connectivity-in-H}, we can say that $s$ and $t$ have 
		at least $k$ vertex-disjoint 
		paths in $H$.
		Deleting $X$ can remove at most $\card{X} \leq k-1$ of these paths in $H-X$.
		This implies that there is still an $s$-$t$ path in $H-X$ and thus there is an edge between $S$ 
		and $T$ in $H-X$, a contradiction. 
		
		\item Case 2: $s$ and $t$ have less than $2k$ vertex-disjoint paths between them in $G$. \\
		Since there are fewer than $2k$ vertex-disjoint paths between $s$ and $t$ in 
		$G$, by conditioning on the event of \Cref{clm:edges-in-H}, $e$ would be preserved in $H$, a contradiction 
		with $H$ having no edge between $S$ and $T$. 
	\end{itemize}	
	In conclusion, we get that $H$ is a certificate of $k$-vertex-connectivity for $G$ with high probability.
\end{proof}

\Cref{thm:k-con-cert} now follows immediately from~\Cref{lem:cert-space} and~\Cref{lem:cert-corr}.

	Before moving on from this section, we present the following corollary 
	of~\Cref{alg:cert} that allows for using this algorithm
	for some other related problems in dynamic streams as well. 
	
	\begin{corollary}\label{cor:k-con-extension}
		The subgraph $H$ output by~\Cref{alg:cert} with high probability satisfies the following guarantees: 
		\begin{enumerate}[label=$(\roman*)$]
			\item For any pair of vertices $s,t$ in $G$, there are at least $k$ vertex-disjoint $s$-$t$ paths in $G$ iff there at least $k$ vertex-disjoint $s$-$t$ paths in $H$ (this holds even if $G$ is not $k$-vertex-connected). 
			\item Every vertex cut of $H$ with size less than $k$ is a vertex cut in $G$ and vice versa (this means all vertex cuts of $G$ are preserved in $H$
			as long as their size is less than $k$). 
		\end{enumerate}
	\end{corollary}
\noindent
The proof of this corollary is identical to that of~\Cref{lem:cert-corr} and is thus omitted. 	

\subsection{Proof of \texorpdfstring{\Cref{clm:connectivity-in-H}}{Lemma} }
We prove \Cref{clm:connectivity-in-H} in this part following the same approach as in \cite{guha2015vertex}. 
For this proof, without loss of generality, we can assume that $k > 1$: for 
$k=1$, each graph $G_i$ is the same as $G$ and thus the algorithm in 
\Cref{prop:spanning-forest} computes an $s$-$t$ path 
which will be added to $H$,  trivially implying the proof.

Fix any pair of vertices $s,t$ with at least $2k$ vertex-disjoint paths between them.
We choose an arbitrary set $X$ of vertices with size $k-1$, and the goal is to show that $s$ and $t$ remain 
connected in the graph $H-X$ with very high probability. 
We do so by showing that out of the at least $k$ vertex-disjoint paths between $s$ and $t$ in $G-X$, with probability $1-n^{-\Theta(k)}$, at least one of them is 
entirely sampled as part of the \emph{subset} of $G_i$'s for $i \in [r]$ that do not contain any vertex from $X$.  
This will be sufficient to prove existence of a $s$-$t$ path in $H-X$.
A union bound over the $\binom{n}{k-1}$ choices of $X$ and $\binom{n}{2}$ pairs $s,t$ 
concludes the proof.

Fix $X$ as a set of $k-1$ vertices that contains neither $s$ nor $t$. Define:
\begin{align}
I(X):=\set{i \in [\reps]: V_i \cap X =\emptyset}; \label{eq:index-set}
\end{align}
that is, the indices of sampled graphs in $G_1,\ldots,G_r$ that contain no vertex from $X$.
We first argue that $\card{I(X)}$ is large with high probability. 

\begin{claim}\label{clm:I(X)}
	$\Pr\paren{\card{I(X)} \leq \reps/8} \leq n^{-5k}$.
\end{claim}
\begin{proof}
	Fix any index $i \in [r]$ and a vertex $v \in X$. The probability that $v$ is not sampled in $V_i$ is $(1-1/k)$ by definition and thus, 
	\[
		\Pr\paren{V_i \cap X = \emptyset} = (1-1/k)^{k-1} \geq 1/4,
	\]
	given that $k > 1$ (as argued earlier) and the choice of vertices is independent in $V_i$. Therefore, we have, 
	\[
	\Exp\card{I(X)} = \reps \cdot (1-1/k)^{k-1} \geq \reps/4.
	\]
	By an application of the Chernoff bound (\Cref{prop:chernoff}) with $\mu_L = \reps/4$ and $\eps=1/2$, we have, 
	\[
		\prob{\card{I(X)} \leq \reps/8} \leq \exp(-\reps/4 \cdot 1/10) < n^{-5k}. \qedhere
	\]	
\end{proof}

	In the rest of the proof we condition on the event that $\card{I(X)} \geq \reps/8$. To continue, we need some definitions. 
	There are more than $k$ vertex-disjoint paths between $s$ and $t$ in $G-X$ since there 
	were $2k$ of them in $G$ and only $k-1$ vertices (set $X$) are deleted. 
	Choose $k$ of them arbitrarily denoted by $P_1(X), \ldots , P_k(X)$.
	For each path $P_j(X)$, let $a_j$ be the edge incident to $s$, $B_j$ be the remaining path until the 
	final edge $c_j$ which is incident to $t$ -- it is possible for $a_i$ and $c_i$ to be the same and $B_i$ be empty (see \Cref{fig:aBc} for an illustration).

	\begin{figure}[h!]
	\centering
	{   \resizebox{150pt}{80pt}{ \input{paths2}}  }
	\caption{An illustration of the $s$-$t$ paths $P_1(X),P_2(X),\ldots,P_k(X)$.
		Each $P_j(X)$ consists of an edge $a_j$ from $s$, a path $B_j$ until the last edge $c_j$ to $t$.
		\label{fig:aBc}}
	\end{figure}
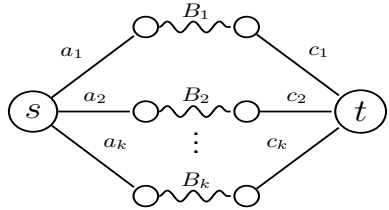

	We define the notion of ``preserving'' a path. 

		\begin{definition}
		Let $G_X := \cup_{i \in I(X)} G_i$ be the union of graphs indexed in~\Cref{eq:index-set}. We say that 
		a path $P$ in $G-X$ is \textnormal{\textbf{preserved}} in $G_X$ iff for every edge $e \in P$, 
		there exists at least one $i \in I(X)$ such that $e \in G_i$; in other words, the entire path $P$ belongs to $G_X$. 
	\end{definition}
		
	We are going to show that with high probability, at least one path $P_j(X)$ for $j \in [k]$ is preserved by $G_X$. Before that, we have the following claim that allows 
	us to use this property to conclude the proof. 
	\begin{claim}\label{clm:st-path}
		If any $s$-$t$ path $P_j(X)$ for $j \in [k]$ is preserved in $G_X$ then $s$ and $t$ are connected in 
		$H-X$. 
	\end{claim}
	\begin{proof}
		Given that $P := P_j(X)$ is preserved, we have that for any edge $e=(u,v) \in P$, there is some graph $G_i$ for $i \in I(X)$ that contains $e$. 
		This means that $u,v$ are connected in $G_i$ which in turn implies that the spanning forest $T_i$ of 
		$G_i$ contains a path between $u$ and $v$. 
		Moreover, since $i \in I(X)$, we know that $G_i$ and hence $T_i$ contain no vertices of $X$ and thus $u$ and $v$ are connected in $T_i - X$ as well. 
		Stitching together these $u$-$v$ paths for every edge $(u,v) \in P$ then gives us a walk between $s$ and $t$ in $H-X$,  implying that $s$ and $t$ are connected in $H-X$. 
	\end{proof}

	We will now prove that some path $P_j(X)$ for $j \in [k]$ is preserved with very high probability.
	\begin{claim}\label{clm:path-prob}
		Conditioned on $\card{I(X)} \geq r/8$, the followings three probabilities are each at most $n^{-2k}$:
		\begin{enumerate}
			\item $\prob{a_j \notin G_X \text{ for at least $k/3$ values of } j \in [k]}$;
			\item $\prob{B_j \not\subseteq G_X \text{ for at least $k/3$ values of } j \in [k]}$;			
			\item $\prob{c_j \notin G_X \text{ for at least $k/3$ values of } j \in [k]}$.
		\end{enumerate}
	\end{claim}
	\begin{proof}	 
		To start the proof, note that even conditioned on a choice of $I(X)$, the vertices in each path $P_j(X)$ appear independently in each graph $G_i$ for $i \in I(X)$. This is because these paths do not intersect with $X$ and 
		by the independence in sampling of each graph $G_i$ for $i \in [r]$.  Moreover, given that these paths are vertex-disjoint (although share $s$ and $t$), 
		the choices of their \emph{inner} vertices across each graph $G_i$ for $i \in [r]$, are independent. We 
		crucially use these properties in this proof. 
	
		An edge is present in $G_i$ if both of its endpoints are sampled which happens with probability 
		$1/k^2$.
		Thus, each edge in $B_j$ is not present in $G_i$ with probability 
		$(1-1/k^2)$ and hence is not 
		present 
		in $G_X$ with probability $(1-1/k^2)^{\card{I(X)}}$. 
		Hence, by the union bound, 
		\[
		\prob{B_j \not\subseteq  G_X} \leq \card{B_j} \cdot \left(1-1/k^2 \right)^{\card{I(X)}} \leq n \cdot 
		\left(1-1/k^2\right)^{r/8} \leq n \cdot \exp\paren{-200k^2\ln{n}/8k^2} = n^{-24}.
		\]
		Finally, note that since the paths $B_j$ for $j \in [k]$ are vertex-disjoint, the probability of the above event is independent for each one. Thus, 
		\[
			\prob{B_j \not\subseteq G_X \text{ for at least $k/3$ values of } j \in [k]} \leq 
			\binom{k}{k/3} \paren{n^{-24}}^{k/3} \leq2^{k} \cdot  n^{-8k} \leq n^{-7k}.
		\]
		As such, the entire path $B_j$ will lie inside $G_X$ for at least $2k/3$ 
		values of $j \in [k]$ with very high probability.
		
		The analysis for $a_j$'s and $c_j$'s is slightly different since one of their endpoints, namely, $s$ and $t$, respectively, is shared across all of them. 
		But the proofs for $a_j$'s and $c_j$'s are entirely symmetric, so we just consider 
		$a_j$'s.
		Consider the set of indices 
		\[
		I_s(X) := I(X)~\cap~\set{i \in [\reps]: s \in G_i};
		\] 
		that is the graphs in $I(X)$ which additionally contain the vertex $s$. 
		For $i \in I_s(X)$, the graph $G_i$ contains the vertex $s$ and but no vertex from $X$.
		We know that $\Exp\card{I_s(X)} = \card{I(X)}/k$ since probability of sampling vertex $s$ in any $G_i$ is 
		$1/k$.
		By an application of the Chernoff bound (\Cref{prop:chernoff}) with $\eps=0.5$, we have, 
		\begin{align*}
			\prob{\card{I_s(X)} \leq \card{I(X)}/2k} \leq  \exp\paren{-\card{I(X)}/10k} \leq \exp(-200k^2\ln 
			n/80k) = n^{-2.5k}. 
		\end{align*} 
		Moreover, for any $i \in I_s(X)$, the probability that $a_j$ is in $G_i$ is $1/k$. Thus, for any fixed $j \in [k]$, 
		\[
			\prob{a_j \notin G_X} = (1-1/k)^{\card{I_s(X)}}.
		\]
		Combining the above two equations, we have, 
		\begin{align*}
			&\prob{a_j \not\in G_X \text{ for at least $k/3$ values of } j \in [k]} \\
			&\hspace{20pt}\leq \Pr\Paren{\card{I_s(X)} \leq \card{I(X)}/2k} + \Pr\Paren{a_j \not\in G_X \text{ for at least $k/3$ values of } j \mid \card{I_s(X)} > \card{I(X)}/2k} \tag{by the law of total probability} \\
			&\hspace{20pt}\leq n^{-2.5k} + \binom{k}{k/3} \paren{(1-1/k)^{\card{I(X)}/2k}}^{k/3} \leq 
			n^{-2.5k} + 2^{k} \cdot \left(\exp\paren{-200k^2\ln{n}/16k^2}\right)^{k/3} < n^{-2k}. 
		\end{align*}
		The same property also holds for $c_j$'s by symmetry, concluding the proof. 
	\end{proof}
	
	By union bound over the events of~\Cref{clm:I(X),clm:path-prob}, we have that there exists an index $j \in [k]$ such that the path $P_j(X)$ is preserved in $G_X$. 
	Thus, by~\Cref{clm:st-path}, for a fixed choice of $X$, and $s,t$, the probability that $s$ and $t$ are 
	not connected in $H-X$
	is at most $4n^{-2k}$. A union bound over the choices of $X$ and $s,t$, then implies that the probability that even one such choice of $X$ and $s,t$ exists is at most 
	\[
		\binom{n}{k-1} \cdot \binom{n}{2} \cdot 4n^{-2k} \leq n^{k+1} \cdot 4n^{-2k} = 4n^{-k+1} <  4n^{-1},
	\]
	since $k \geq 2$. This completes the 
	proof of~\Cref{clm:connectivity-in-H}.

\subsection{Proof of \texorpdfstring{\Cref{clm:edges-in-H}}{Lemma} }
	We now prove \Cref{clm:edges-in-H}.
	For this proof also, without loss of generality, we can assume that $k > 1$: for 
	$k=1$, each graph $G_i$ is the same as $G$ and thus the algorithm in 
	\Cref{prop:spanning-forest} computes the only $s$-$t$ path, namely, the edge $(s,t)$ (as $s$ and $t$ can only be $1$-connected through the edge $(s,t)$)
	which will be added to $H$, thus trivially implying the proof.  We now consider the main case. 
	
	Fix any pair of vertices $s,t \in G$ which have less than $2k$ vertex-disjoint paths between them. We know 
	that deleting the 
	edge $(s,t)$ and some set of vertices $X$ of size less than $2k$ should disconnect $s$ and $t$.
	For any $i \in [r]$, we call the graph $G_i$ \textbf{good} if it samples both $s$ and $t$ and does not sample any 
	vertex from $X$. See~\Cref{fig:partition_LowCon} for an illustration.

	\begin{figure}[ht!]
	\centering
	{   \resizebox{170pt}{110pt}{ \input{partition_fill}}  }
	\caption{An illustration of a good graph $G_{i^*}$ wherein both vertices $s$ and $t$ are sampled 
	and all the vertices in set $X$ are not. Thus, none of  the $s$-$t$ paths, except for the edge $e$, 
	exist in $G_{i^*}$ since they all pass through $X$. Therefore, the spanning forest $T_{i^*}$ 
	necessarily contains the edge $e=(s,t)$. 
		\label{fig:partition_LowCon}}
\end{figure}
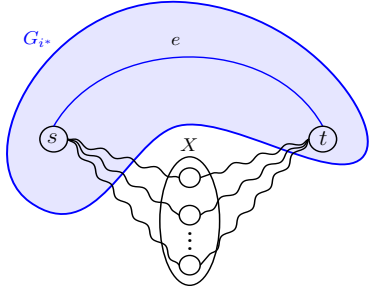

	We have, 
	\[
		\prob{G_i \text{ is good}} = 1/k^2 \cdot (1-1/k)^{2k-1} \geq 1/8k^2 \tag{as $k\geq 2$ so $(1-1/k)^{2k-1} \geq (1/2)^{3}$}.
	\]
	Given the independence of choices of $G_i$ for $i \in [r]$, we have, 
	\[
		\prob{\text{No $G_i$ is good}} \leq (1-1/8k^2)^{r} \leq \exp\paren{-200 k^2 \ln n/8k^2} = n^{-25}.
	\]
	
	Therefore, there is a graph $G_{i^*}$ for $i^* \in [r]$ where $s$ and $t$ are sampled but $X$ is 
	not (see \Cref{fig:partition_LowCon}).
	This means that the spanning forest $T_{i^*}$ has to contain the edge $(s,t)$ as 
	there is no other 
	path between $s$ and $t$ (we have effectively ``deleted'' $X$ by not sampling it).
	Thus, the edge $(s,t)$ belongs to $H$ with probability at least $1-n^{-25}$. 
	A union bound over all possible pairs $s,t \in G$ concludes the proof.

%% file: partition.tex
\tikzset{every picture/.style={line width=0.75pt}} %set default line width to 0.75pt        

\begin{tikzpicture}[x=0.75pt,y=0.75pt,yscale=-1,xscale=1]
	
	%Shape: Circle [id:dp5245205304598188] 
	\draw   (121,179) .. controls (121,149.18) and (145.18,125) .. (175,125) .. controls (204.82,125) and 
	(229,149.18) .. (229,179) .. controls (229,208.82) and (204.82,233) .. (175,233) .. controls 
	(145.18,233) and (121,208.82) .. (121,179) -- cycle ;
	%Shape: Circle [id:dp6155603624545998] 
	\draw   (332,179) .. controls (332,149.18) and (356.18,125) .. (386,125) .. controls (415.82,125) 
	and (440,149.18) .. (440,179) .. controls (440,208.82) and (415.82,233) .. (386,233) .. controls 
	(356.18,233) and (332,208.82) .. (332,179) -- cycle ;
	%Shape: Circle [id:dp2171011571335204] 
	\draw   (183,146.5) .. controls (183,142.63) and (186.13,139.5) .. (190,139.5) .. controls 
	(193.87,139.5) and (197,142.63) .. (197,146.5) .. controls (197,150.37) and (193.87,153.5) .. 
	(190,153.5) .. controls (186.13,153.5) and (183,150.37) .. (183,146.5) -- cycle ;
	%Shape: Circle [id:dp7125339151364563] 
	\draw   (364,146.5) .. controls (364,142.63) and (367.13,139.5) .. (371,139.5) .. controls 
	(374.87,139.5) and (378,142.63) .. (378,146.5) .. controls (378,150.37) and (374.87,153.5) .. 
	(371,153.5) .. controls (367.13,153.5) and (364,150.37) .. (364,146.5) -- cycle ;
	%Shape: Circle [id:dp6182526352958269] 
	\draw   (246,203) .. controls (246,183.67) and (261.67,168) .. (281,168) .. controls (300.33,168) 
	and (316,183.67) .. (316,203) .. controls (316,222.33) and (300.33,238) .. (281,238) .. controls 
	(261.67,238) and (246,222.33) .. (246,203) -- cycle ;
	%Curve Lines [id:da03595061358336027] 
	\draw [draw=blue]   (190,139.5) .. controls (225,74) and (336,72) .. (371,139.5) ;
	
	% Text Node
	\draw (145,99) node [anchor=north west][inner sep=0.75pt]   [align=left] {\large{$S$}};
	% Text Node
	\draw (404,99) node [anchor=north west][inner sep=0.75pt]   [align=left] {\large{$T$}};
	% Text Node
	\draw (383,138) node [anchor=north west][inner sep=0.75pt]   [align=left] {\LARGE{$t$}};
	% Text Node
	\draw (168,138) node [anchor=north west][inner sep=0.75pt]   [align=left] {\LARGE{$s$}};
	% Text Node
	\draw (276,145) node [anchor=north west][inner sep=0.75pt]   [align=left] {\large{$X$}};
	% Text Node
	\draw (242,65.5) node [anchor=north west][inner sep=0.75pt]   [align=left] {\Large{$e \in G-H$}};

\end{tikzpicture}

%% file: paths2.tex
%\begin {tikzpicture}[-latex ,auto ,node distance =12 mm and 20 mm ,on grid, semithick ,
%state/.style ={ circle ,top color =white , bottom color = white , draw,black , text=black , minimum 
%width =5 mm},
%labelstate/.style ={ circle ,top color =white , bottom color = white , draw,white , text=black , 
%minimum 
%	width =1 cm}
%]
%
%\node[state] (S) {$s$};
%\node[state] (V1) [right =of S] {};
%\node[labelstate] (V2) [right =of V1] {};
%\node[state] (V3) [right =of V2] {};
%\node[state] (T) [right=of V3] {$t$};
%
%\node[state] (U1) [above right =of S] {};
%\node[labelstate] (U2) [right =of U1] {};
%\node[state] (U3) [right =of U2] {};
%
%\node[state] (W1) [below right =of S] {};
%\node[labelstate] (W2) [right =of W1] {};
%\node[state] (W3) [right =of W2] {};
%
%\path[every node/.style={font=\sffamily\small}]
%(S) edge [bend right = 0] node {} (V1)
%(V1) edge [bend  left] node {} (V2);
%\end{tikzpicture}

\begin{tikzpicture}[
	shorten >=1pt, auto, thick,
	node distance=1cm,
	state/.style={circle,draw,fill=white,font=\sffamily\Large\bfseries},
	labelstate/.style={circle,draw,fill=white,  draw,white, text=black, font=\sffamily\Large\bfseries}
	]
\node[state] (S) {$s$};
\node[state] (V1) [right =of S] {};
%\node[labelstate] (V2) [right =of V1] {$\ldots$};
\node[state] (V3) [right =of V1] {};
\node[state] (T) [right=of V3] {$t$};

\node[state] (U1) [above =of V1] {};
%\node[labelstate] (U2) [above =of V2] {$\ldots$};
\node[state] (U3) [above =of V3] {};

\node[state] (W1) [below =of V1] {};
%\node[labelstate] (W2) [below =of V2] {$\ldots$};
\node[state] (W3) [below =of V3] {};

\path[every node/.style={font=\sffamily\small}]
(S) edge [bend right = 0] node {$a_2$} (V1)
(V1) edge [bend  right = 0, decorate,decoration=snake] node {$B_2$} node [below] 
{\textbf{$\vdots$}} (V3)
%(V2) edge [bend  right = 0] node {$B_2$} (V3)
(V3) edge [bend  right = 0] node {$c_2$} (T)

(S) edge [bend right = 0] node {$a_1$} (U1)
(U1) edge [bend  right = 0, decorate,decoration=snake] node {$B_1$} (U3)
%(U2) edge [bend  right = 0] node {$B_1$} (U3)
(U3) edge [bend  right = 0] node {$c_1$} (T)

(S) edge [bend right = 0] node {$a_k$} (W1)
(W1) edge [bend  right = 0, decorate,decoration=snake] node {$B_k$} (W3)
%(W2) edge [bend  right = 0] node {$B_k$} (W3)
(W3) edge [bend  right = 0] node {$c_k$} (T)
;
\end{tikzpicture}

%% file: partition_fill.tex
\tikzset{every picture/.style={line width=0.75pt}} %set default line width to 0.75pt        

\begin{tikzpicture}[x=0.75pt,y=0.75pt,yscale=-1,xscale=1]
	
	%Shape: Circle [id:dp2171011571335204] 
	\draw   (180.75,148.75) .. controls (180.75,143.64) and (184.89,139.5) .. (190,139.5) .. controls 
	(195.11,139.5) and (199.25,143.64) .. (199.25,148.75) .. controls (199.25,153.86) and (195.11,158) 
	.. (190,158) .. controls (184.89,158) and (180.75,153.86) .. (180.75,148.75) -- cycle ;
	%Curve Lines [id:da03595061358336027] 
	\draw  [draw=blue]  (190,139.5) .. controls (225,74) and (335,72) .. (370,139.5) ;
	%Shape: Circle [id:dp6494366776183813] 
	\draw   (274,176.5) .. controls (274,172.63) and (277.13,169.5) .. (281,169.5) .. controls 
	(284.87,169.5) and (288,172.63) .. (288,176.5) .. controls (288,180.37) and (284.87,183.5) .. 
	(281,183.5) .. controls (277.13,183.5) and (274,180.37) .. (274,176.5) -- cycle ;
	%Shape: Circle [id:dp3026905736811003] 
	\draw   (274,203.5) .. controls (274,199.63) and (277.13,196.5) .. (281,196.5) .. controls 
	(284.87,196.5) and (288,199.63) .. (288,203.5) .. controls (288,207.37) and (284.87,210.5) .. 
	(281,210.5) .. controls (277.13,210.5) and (274,207.37) .. (274,203.5) -- cycle ;
	
	%Shape: Polygon Curved [id:ds32399922720428487] 
	\begin{scope}[on background layer]
	\draw [draw=blue, fill=blue!10, line width=1pt]  (274,50) .. controls (361,54) and (416,142.5) .. 
	(396,162.5) .. controls (376,182.5) and 
	(303,128) .. (271,140) .. controls (239,152) and (218,228) .. (174,194) .. controls (130,160) and 
	(187,46) .. (274,50) -- cycle ;
	\end{scope}
	
	%Shape: Ellipse [id:dp9308794421428075] 
	\draw   (281,161.5) .. controls (292.05,161.5) and (301,182.21) .. (301,207.75) .. controls 
	(301,233.29) and (292.05,254) .. (281,254) .. controls (269.95,254) and (261,233.29) .. 
	(261,207.75) .. controls (261,182.21) and (269.95,161.5) .. (281,161.5) -- cycle ;
	%Shape: Circle [id:dp09414894967698273] 
	\draw   (274,239.5) .. controls (274,235.63) and (277.13,232.5) .. (281,232.5) .. controls 
	(284.87,232.5) and (288,235.63) .. (288,239.5) .. controls (288,243.37) and (284.87,246.5) .. 
	(281,246.5) .. controls (277.13,246.5) and (274,243.37) .. (274,239.5) -- cycle ;
	%Shape: Circle [id:dp3099843839348442] 
	\draw   (360.75,148.75) .. controls (360.75,143.64) and (364.89,139.5) .. (370,139.5) .. controls 
	(375.11,139.5) and (379.25,143.64) .. (379.25,148.75) .. controls (379.25,153.86) and 
	(375.11,158) .. (370,158) .. controls (364.89,158) and (360.75,153.86) .. (360.75,148.75) -- cycle ;
	%Straight Lines [id:da20700725633453687] 
	\draw  [line width=0.75pt, decorate,decoration={snake, segment length=5mm, amplitude=0.5mm}]  
	(199.25,148.75) 
	-- (274,176.5) ;
	%Straight Lines [id:da5951967359992445] 
	\draw   [line width=0.75pt, decorate,decoration={snake, segment length=5mm, amplitude=0.5mm}] 
	(199.25,148.75) -- (274,203.5) ;
	%Straight Lines [id:da2784350079649571] 
	\draw  [line width=0.75pt, decorate,decoration={snake, segment length=5mm, amplitude=0.5mm}]  
	(199.25,148.75) -- (274,239.5) ;
	%Straight Lines [id:da5992291561738305] 
	\draw   [line width=0.75pt, decorate,decoration={snake, segment length=5mm, 
	amplitude=0.5mm}]  
	(360.75,148.75) --  (288,176.5);
	%Straight Lines [id:da5417599795852959] 
	\draw  [line width=0.75pt, decorate,decoration={snake, segment length=5mm, amplitude=0.5mm}]  
	 (360.75,148.75) --  (288,203.5);
	%Straight Lines [id:da31878764126555503] 
	\draw  [line width=0.75pt, decorate,decoration={snake, segment length=5mm, amplitude=0.5mm}]  
	 (360.75,148.75) --  (288,239.5);
	
	% Text Node
	\draw (168,70) node [anchor=north west][inner sep=0.75pt]   [align=left] 
	{\textcolor{blue}{$G_{i^*}$}};
	% Text Node
	\draw (366,142) node [anchor=north west][inner sep=0.75pt]   [align=left] {\large{$t$}};
	% Text Node
	\draw (184,143) node [anchor=north west][inner sep=0.75pt]   [align=left] {\large{$s$}};
	% Text Node
	\draw (273,147) node [anchor=north west][inner sep=0.75pt]   [align=left] {$X$};
	% Text Node
	\draw (267,74) node [anchor=north west][inner sep=0.75pt]   [align=left] {$e$};
	
	% Text Node (dots)
	\draw (277.5,206) node [anchor=north west][inner sep=0.75pt]   [align=left] {\textbf{$\vdots$}};
	
	%snake lines:  [decorate,decoration=snake]
	% {snake, segment length=5mm, amplitude=15mm}
	
\end{tikzpicture}

%% file: alg.tex
% !TeX root = main.tex 
%!TEX root = main.tex

\section{The Dynamic Streaming Algorithm}\label{sec:Algorithm}

We  present our single pass dynamic streaming algorithm for $k$-vertex-connectivity in this section. 
The algorithm outputs a certificate of $k$-vertex-connectivity for the input graph at the end of the 
stream. 
Thus, by the definition of the certificate, to know whether or not the input graph is 
$k$-vertex-connected, it suffices to test if the certificate is $k$-vertex-connected, which
can be done at the end of the stream using any offline algorithm. 
The following theorem formalizes \Cref{res:alg}. 
\begin{theorem}\label{thm:one-pass-con-alg}
	There is a randomized dynamic streaming algorithm that given an integer $k\geq 1$ 
	before the stream and a graph $G=(V,E)$ in the stream, outputs a certificate $H$ of 
	$k$-vertex-connectivity of $G$ with high probability using $O(kn \cdot \log^4 n)$ bits of 
	space.
\end{theorem}

This algorithm is just an implementation of \Cref{alg:cert} in dynamic streams.
We fix the vertex sets $V_i$ in \Cref{alg:cert} before the stream so the only thing we need to 
specify is how we compute the spanning forests during the stream.
We compute a spanning forest $T_i$ of $G_i$ in the stream using the dynamic streaming 
algorithm in \Cref{prop:spanning-forest} with parameters $N= \card{V_i}$ and $\delta=n^{-4}$.
After the stream, we output the certificate $H$.
This completes the description of the streaming algorithm.

We start by bounding the space of this algorithm.

\begin{lemma}\label{lem:one-pass-space}
	This algorithm uses $O(k n \cdot \log^4 n)$ bits of space with high probability.
\end{lemma}
\begin{proof}
	During the stream, we run a spanning forest algorithm for each graph $G_i$ for $i \in [r]$. 
	The algorithm of~\Cref{prop:spanning-forest} with parameters $N= \card{V_i}$ 
	and $\delta=n^{-4}$ takes at most $c \card{V_i} \cdot \log^3 (n^4 \card{V_i})$ bits of space 
	for	some absolute constant $c$.
	We store $\reps=O(k^2 \ln n)$ spanning forests so the total space taken is
	\[
		\sum_{i=1}^{r} c\card{V_i} \cdot \log^3 (n^4 \card{V_i})
		\leq c \log^3 (n^5) \sum_{i=1}^{r} \card{V_i}
		=  O(kn \log^4 n),
	\]	
	where the first inequality uses $\card{V_i} \leq n$ and the second one uses $ \sum_{i=1}^{r} 
	\card{V_i}=  O(kn \log n)$ (by \Cref{clm:cert-space-help}).
\end{proof}
	
	We are now ready to prove \Cref{thm:one-pass-con-alg}.
	\begin{proof}[Proof of \Cref{thm:one-pass-con-alg}]
		By~\Cref{lem:one-pass-space}, this algorithm uses $O(k n \cdot \log^4 n)$ bits of 
		space with high probability. 
		The high probability guarantee can be even moved from the space bound to the correctness in the following way: if the space of the algorithms at any point increases 
		beyond this high-probability bound guarantee, we simply terminate the algorithm and output ``fail''. 
		This happens with negligible probability by \Cref{lem:one-pass-space}.
		The algorithm is also correct with high probability by \Cref{thm:k-con-cert}. 
		Moreover, none of the spanning forest algorithms of \Cref{prop:spanning-forest} fail with 
		high probability (by a union bound over the failure probabilities of 
		the $r$ spanning forest algorithms).
		Thus, 
		the streaming algorithm (deterministically) uses $O(kn\cdot\log^4 n)$ bits of space and, by union bound, with high probability outputs a certificate 
		of $k$-vertex-connectivity. 
	\end{proof}

%% file: LowerBound.tex
% !TeX root = main.tex 
%!TEX root = main.tex

\section{The Lower Bound}\label{sec:LowerBound}

In this section, we extend the prior  single-pass lower bound of~\cite{sun2015tight} for vertex connectivity to multi-pass algorithms. The following theorem
formalizes~\Cref{res:lower}. 
\begin{theorem}\label{thm:LB-multipass}
	For any integer $p \geq 1$, any randomized $p$-pass insertion only streaming algorithm that given an integer 
	$1\leq k\leq n/2$ before the stream and an $n$-vertex (multi-)graph $G=(V,E)$ in the stream, outputs 
	whether $G$ is $k$-vertex connected with probability at least $\nicefrac23$, needs $\Omega(kn/p)$ bits of 
	space.
\end{theorem}
%We note that this lower bound holds for every $n$ and $k$ such that $k\leq n/2$.
We use the standard approach of proving lower bounds on the space of streaming algorithms via 
communication complexity (this is further spelled out in the proof of~\Cref{thm:LB-multipass}). 
The communication lower bound itself is proven using a reduction from the well-known set disjointness problem defined as follows.  

\begin{definition}[\textbf{Set Disjointness (\textnormal{\disj}) }] \label{prob:SetDisj}
	For any integer $N \geq 1$, \disj is defined as follows:
	Alice and Bob are given length $N$ binary strings $x \in \set{0,1}^N$ and $y \in \set{0,1}^N$, respectively. 
	They can communicate back and forth and need to output ``No'' if there exists an index $i \in [N]$ such that 
	$x_i=y_i=1$ and ``Yes'' otherwise.
	We assume both players have access to a shared source of randomness.
\end{definition}

We use the following standard lower bound on the communication complexity of this problem.
\begin{proposition}[\!\!\cite{kalyanasundaram1992probabilistic, razborov1990distributional, 
bar2004information}] 
\label{prop:Disj-LB}
	For any integer $N \geq 1$, any two-way randomized protocol for $\disj$ that errs with probability at 
	most $\nicefrac{1}{3}$ needs $\Omega(N)$ bits of communication.
\end{proposition}

We use this result to prove a communication complexity lower bound 
for vertex connectivity. 

\begin{proposition}\label{prop:cc-lower}
	For any integers $n,k \geq 1$ such that $1\leq k\leq n/2$ the following is true. Any randomized communication protocol wherein Alice and Bob receive edges of an $n$-vertex
	(multi-)graph $G=(V,E)$ partitioned between the two, and can output whether or not $G$ is $k$-vertex-connected with probability at least $\nicefrac23$ requires $\Omega(kn)$ bits of communication. 
\end{proposition}
\begin{proof}
	We start with a high level sketch of the proof. 
	We use a reduction from the \disj communication problem for $N=\Theta(kn)$. 
Alice and Bob construct a bipartite graph $G$ on $n$ fixed vertices and pick their edges based on the values in their input strings $x$ and $y$ in $\disj$.
$G$ will be constructed in a way that if $x$ and $y$ are disjoint, then $G$ will contain a complete bipartite graph and has vertex connectivity 
$k$; otherwise, at least one edge is missing and the graph has vertex connectivity strictly less than $k$.
Thus, solving $k$-vertex connectivity also solves \disj implying the space lower bound. We now formalize this idea. 

To prove the bound for parameters $n$ and $k$, we start with an instance of \disj such that $N= k \cdot (n-k)$.
Alice and Bob construct an $n$-vertex bipartite graph $G=(L\sqcup R,E)$ with $k$ vertices 
on $L$ and $n-k$ vertices on $R$ as follows: 
\begin{itemize}
\item \emph{Vertices:} the vertices in $L$ are $u_1,u_2,\ldots,u_k$ and the vertices in $R$ are 
$v_1,v_2,\ldots,v_{n-k}$. 
\item \emph{Edges:} the indices of Alice's string $x$ and Bob's string $y$ can be expressed using coordinates $i\in[k]$ 
and $j \in [n-k]$ (since $N=k \cdot (n-k)$).
If $x_{i,j}=0$ then Alice has an edge $(u_i,v_j)$ and if $y_{i,j}=0$ then Bob has an edge 
$(u_i,v_j)$ (this way, there can be up to two edges between any pairs of vertices). 
\end{itemize}

The following claim is the key part for establishing the correctness of our reduction.
\begin{claim}\label{clm:LB-reduction}
	$G$ is $k$-vertex connected iff $x$ and $y$ are disjoint.
\end{claim}
\begin{proof}
	If $x$ and $y$ are disjoint, then for every $i \in [k],j \in [n-k]$ either $x_{i,j}=0$ or $y_{i,j}=0$ and thus edge 
	$(u_i,v_j)$ exists in $G$.
	Thus, $G$ contains a complete bipartite graph.
	Deleting any set $X$ of $k-1$ vertices leave at least one vertex $u_i \in L$ and one vertex $v_j \in R$. Since 
	$u_i$ and $v_j$ are connected, and any vertex in $L$ is connected to $v_j$ and any vertex in $R$ 
	is connected to $u_i$, we have that $G-X$ is connected.
	Therefore, $G$ is $k$-vertex-connected.

	If $x$ and $y$ are not disjoint, then there are indices $i^*$ and $j^*$ such that $x_{i^*,j^*}=1$ 
	and $y_{i^*,j^*}=1$ implying that edge $(u_{i^*},v_{j^*})$ does not exist in $G$.
	Deleting all vertices in $L$ except $u_{i^*}$ disconnects $v_{j^*}$ from the rest of the graph.
	Thus, there is a vertex cut of size $k-1$  implying that $G$ is not $k$-vertex connected. $\qed_{\,\,\textnormal{\Cref{clm:LB-reduction}}}$
	
\end{proof}

The proof of~\Cref{prop:cc-lower} now follows from~\Cref{prop:Disj-LB}: Alice and Bob, given any instance $(x,y)$ of $\disj$, can construct the graph $G$ in the  reduction without any communication 
and run the protocol for $k$-vertex-connectivity on $G$. If the protocol returns $G$ is $k$-vertex-connected, they return ``Yes'' and otherwise they return ``No''. 
The correctness follows from~\Cref{clm:LB-reduction}. This implies that the $k$-vertex-connectivity protocol  needs $\Omega(N) = \Omega(kn)$ communication by~\Cref{prop:Disj-LB}, concluding the proof. 
\end{proof}

We can now obtain~\Cref{thm:LB-multipass} as a standard corollary of~\Cref{prop:cc-lower}. 

\begin{proof}[Proof of~\Cref{thm:LB-multipass}]
Given a $p$-pass streaming algorithm for the $k$-vertex-connectivity problem, Alice and Bob can use the algorithm to solve the communication 
problem as follows. Alice treats her edges in the communication problem as the first part of the stream and Bob treats his edges as the second part of the stream. 
The players run the streaming algorithm on this stream by communicating the memory content whenever they finish running that pass of the algorithm on their input. 
This requires sending the memory content for $2p-1$ times until Bob can compute the answer of the streaming algorithm. 

Assuming we start with a $p$-pass streaming algorithm that uses only $o(k n/p)$ bits of space, the above approach gives us a 
communication protocol with $o(kn)$ communication for $k$-vertex-connectivity, with the same probability of success as the streaming algorithm. 
This contradicts~\Cref{prop:cc-lower}, and concludes the proof of \Cref{thm:LB-multipass}.
\end{proof}

%% file: ins_only.tex
% !TeX root = main.tex 
%!TEX root = main.tex

\section{An Insertion-Only Streaming Algorithm}\label{app:insertion-only}
We describe the insertion-only streaming algorithm for $k$-vertex-connectivity here. This algorithm has been folklore in the literature already since the introduction of graph streaming model in~\cite{FeigenbaumKMSZ05} (but we are not aware of any explicit reference for this result). 
Several references in the past attribute the algorithm to 
\cite{eppstein1997sparsification} but in fact, it appears that
\cite{cheriyan1993scan} also has almost the complete proof which they present as an online 
algorithm that outputs a certificate of vertex-connectivity (as in~\Cref{def:certificate})\footnote{To the best of our knowledge, the first version of~\cite{cheriyan1993scan} is a technical report in 1991 which predates the conference version of~\cite{eppstein1997sparsification} from 1992.}.
We provide the algorithm and its analysis in the insertion-only streaming model in this appendix for completeness. 

\paragraph{Preliminaries.} Before we present the algorithm we mention two important propositions which will be useful in the 
analysis of the algorithm.
The first is Menger's theorem which gives an equivalent definition of $k$-vertex-connectivity via vertex-disjoint paths. 
\begin{proposition}[Menger's Theorem; c.f.~{\cite[Theorem 17]{west2001introduction} 
}]\label{prop:menger}
	Let $G$ be an undirected graph and $s$ and $t$ be two non-adjacent vertices. 
	Then the size of the minimum vertex cut for $s$ and $t$ is equal to the maximum 
	number of vertex-disjoint paths between $s$ and $t$. \\
	Moreover, a graph is $k$-vertex-connected if and only if every pair of vertices has at 
	least $k$ vertex-disjoint paths in between.
\end{proposition}

The next  is Mader's theorem on existence of $k$-vertex-connected subgraphs on sufficiently dense graphs. 
\begin{restatable}[Mader's Theorem; c.f.~{\cite[Theorem 1.4.3]{diestel2005graph}} 
]{proposition}{Mader}\label{prop:mader-thm}
	For any $k> 1$, if an undirected graph has at least $2k-1$ vertices and at least 
	$(2k-3)(n-k+1)+1$ edges,  it contains a $k$-vertex-connected subgraph.
\end{restatable}
Unlike Menger's theorem, there are not many sources that contain a complete proof of Mader's theorem in the above formulation and with the given parameters (despite being a well-known result mentioned in various sources, e.g., with a 
different formulation in~{\cite[Theorem 1.4.3]{diestel2005graph}}). Thus, we also present a simple proof of this theorem in~\Cref{app:Mader} for interested readers.

\paragraph{The insertion-only streaming algorithm.} We will reprove the following folklore theorem. 
\begin{theorem}[cf.~\cite{cheriyan1993scan,eppstein1997sparsification,
		sun2015tight,guha2015vertex}]\label{thm:ins-only}
	There is a deterministic insertion-only streaming algorithm that given an integer $k\geq 1$ 
	before the stream and a graph $G=(V,E)$ in the stream, outputs a certificate $H$ of 
	$k$-vertex-connectivity of $G$ using $O(kn \log n)$ bits of space.
\end{theorem}

The algorithm is very simple: when an edge $(u,v) $ arrives in the stream, store the edge if and only if the number of vertex-disjoint paths 
between $u$ and $v$ is less than $k$.

\begin{Algorithm}\label{alg:ins-only-k-con}
	An insertion-only streaming algorithm for $k$-vertex connectivity.
	
	\medskip
	
	\textbf{Input:} A graph $G = (V,E)$ specified in a stream and an integer $k$ specified at 
	the beginning of the stream.
	
	\medskip
	
	\textbf{Output:} A certificate $H$ for $k$-vertex connectivity of $G$.
	
	\smallskip
	\begin{enumerate}
		\item Let $F=\emptyset$. When any edge $e=(u,v)$ arrives, if the maximum number of vertex-disjoint 
		paths between $u$ and $v$ in $(V, F)$ is less than $k$, update $F \leftarrow F \cup \set{e}$ (otherwise, $F$ remains unchanged).
		\item When the stream ends output $H:= (V,F)$ as a certificate for $k$-vertex connectivity of 
		$G$.
	\end{enumerate}
\end{Algorithm}

We start by proving the correctness of the algorithm.
\begin{lemma}\label{lem:ins-only-corr}
	Subgraph $H$ output by~\Cref{alg:ins-only-k-con} is a certificate for $k$-vertex-connectivity of $G$.
\end{lemma}
\begin{proof}	
	We need to show that $H$ is $k$-vertex-connected iff $G$ is $k$-vertex-connected.
	If $H$ is $k$-vertex-connected then $G$ is $k$-vertex-connected since $H$ is a 
	subgraph of $G$ on the same set of vertices. 
	
	Suppose now towards a contradiction that $G$ is $k$-vertex connected, but $H$ is not.
	This means that there is a vertex cut $X$ of size at most $k-1$ such that $S,X,T$ is a partition of 
	$V$ and there are no edges between $S$ and $T$.
	Since $G$ is $k$-vertex-connected, it has an edge $e=(u,v)$ between $S$ and $T$.
	Edge $e$ is not stored in $H$ by \Cref{alg:ins-only-k-con} and so $u$ and $v$ have at least $k$ 
	vertex disjoint paths between them.
	But this means that deleting $X$, a set of at most $k-1$ vertices, cannot disconnect $S$ and $T$, leading to a contradiction.
\end{proof}

We now prove the following claim which will be helpful in proving the space bound.
\begin{claim}\label{clm:no-high-con-subgraph}
	The certificate $H$ of~\Cref{alg:ins-only-k-con} does not contain any subgraph that is $(k + 1)$-vertex connected.
\end{claim}
\begin{proof}
	Assume for contradiction that $H$ contains a subgraph $J$ that is $(k + 1)$-vertex connected.
	Let $e=(u,v)$ be the last edge added to $J$ by \Cref{alg:ins-only-k-con}.
	By~\Cref{prop:menger}, this means $u$ and $v$ have at least $k+1$ vertex-disjoint paths 
	between them in $J$ 
	and thus have at least $k$ vertex disjoint paths between them in $J-\set{e}$.
	Therefore, when $e$ arrives in the stream, it is not stored since $u$ and $v$ already have $k$ 
	vertex disjoint paths between them in $H-\set{e}$. But this is a contradiction with $e$ being in $H$. 
\end{proof}

Finally, we prove that $H$ contains at most $2kn$ edges.
\begin{lemma}\label{lem:ins-only-space}
	The certificate $H$ of~\Cref{alg:ins-only-k-con} contains at most $2kn$ edges.
\end{lemma}
\begin{proof}
	If $n < 2k-1$ then $H$ contains at most $n (n-1)/2 \leq 2kn$ edges proving the claim.
	Thus, consider $n\geq 2k-1$.
	If $H$ has more than $2kn$ edges then by \Cref{prop:mader-thm} it contains a $(k + 1)$-vertex 
	connected subgraph. 
	But $H$ cannot contain any subgraph that is $(k + 1)$-vertex connected 
	by \Cref{clm:no-high-con-subgraph}.
\end{proof}

We can now conclude the proof of \Cref{thm:ins-only}.
\begin{proof}[Proof of \Cref{thm:ins-only}]
	\Cref{lem:ins-only-corr} proves that $H$ is a certificate for $k$-vertex-connectivity of $G$.
	\Cref{lem:ins-only-space} proves that $H$ contains at most $2kn$ edges implying that 
	\Cref{alg:ins-only-k-con} uses $O(k n \log n)$ bits of space.
\end{proof}

%% file: mader.tex
\section{Mader's Theorem}\label{app:Mader}

We present a self-contained proof of Mader's theorem in this section for the interested reader. 
Consider the following restatement of the proposition.
\Mader*

\begin{proof}
	We fix a value of $k$ and prove the proposition by induction on $n$, the number of vertices. Our induction hypothesis is as follows: For any $t\geq 2k$, if an undirected graph 
	has $t-1$ vertices and at least $(2k-3)(t-k)+1$ edges then it contains a $k$-vertex connected subgraph.

	\noindent\textbf{Base case:} when $t=2k$. 
	
	We have $m\geq (2k-3)(k)+1 = 2k^2-3k+1$.
	A clique on $2k-1$ vertices has $(2k-1)(2k-2)/2 = 2k^2-3k+1$ edges.
	Thus, the only graph on $2k-1$ vertices that satisfies the edge lower bound is a clique that is 
	$k$-vertex-connected and thus has subsets that are $k$-vertex-connected.
	
	\noindent\textbf{Induction step:} We assume the hypothesis for  integers up to $t$ and prove it for $t+1$, that is if 
	an undirected graph has $t$ vertices and at least $(2k-3)(t-k+1)+1$ edges then it contains a $k$-vertex-connected subgraph. 
	
	Assume towards a contradiction that there is a graph $G$ with $t$ vertices and at least 
	$(2k-3)(t-k+1)+1$ edges which contains no $k$-vertex connected subgraph.
	We first show that $G$ has a large minimum degree.
	\begin{claim}\label{clm:min-degree}
		$G$ has minimum degree $\delta\geq 2k-2$. 
	\end{claim}
	\begin{proof}
		Consider a vertex $v$ with minimum degree $\delta$. 
		Removing $v$ leaves the graph with $t-1 \geq 2k-1$ vertices and $m' 
		\geq(2k-3)(t-k+1)+1-\delta$ 
		edges. 
		If $m' \geq (2k-3)(t-k)+1$ then $G$ contains a $k$-vertex-connected subgraph by induction; thus, we need to have $\delta \geq 2k-2$.
	\end{proof}
	
	We know that $G$ is not $k$-vertex-connected which implies there is a vertex cut $X$ with at most 
	$k-1$ vertices which when deleted disconnects $G$ into components $S$ and $T:=V-X-S$.
	By~\Cref{clm:min-degree}, for any vertex $u\in S$, $\deg(u) \geq 2k-2$.
	Moreover, since there are no edges between $S$ and $T$, any vertex $u \in S$ has neighbors only in $X$ and $S$. Thus, since $\card{X} < k$, we need $S$ to have at least $k-1$ vertices other than 
	$u$ to satisfy the degree requirement of $u$, which implies $\card{S}\geq k$.
	By symmetry, we also have $\card{T}\geq k$. 
	
	Let $G_1$ be the induced subgraph of $G$ on $S\cup X$ with $n_1$ vertices and let $G_2$ be 
	the induced subgraph of $G$ on $T\cup X$ with $n_2$ vertices.
	 Both $G_1$ and $G_2$ do not contain any $k$-vertex-connected subgraphs and have at least 
	 $2k-1$ vertices, so they have strictly fewer than $(2k-3)(n_1 - k+1)+1$ and $(2k-3)(n_2 - 
	 k+1)+1$ 
	 edges, respectively.
	 We now sum the number of edges $m_1$ of $G_1$ and $m_2$ of $G_2$:
	 \begin{align*}
	 	m_1+m_2 &\leq (2k-3)(n_1 - k+1) + (2k-3)(n_2 - k+1) \\
	 	&= (2k-3)(n_1+n_2 -2k+2) \\
	 	&\leq(2k-3)(t -k+1) \\
	 	&< m \tag{Since $m\geq (2k-3)(t -k+1)+1$}
	 \end{align*}
	 But we know $m_1+m_2 \geq m$ because $G_1$ and $G_2$ cover all edges of $G$ (and can even 
	 over count some edges, namely, those with both endpoints in $X$).
	 Thus, we arrive at a contradiction and such a graph $G$ cannot exist.
	 Therefore, we have shown the induction step and proved the proposition.
\end{proof}